\documentclass[sigconf]{acmart}
\AtBeginDocument{%
  \providecommand\BibTeX{{%
    \normalfont B\kern-0.5em{\scshape i\kern-0.25em b}\kern-0.8em\TeX}}}

\setcopyright{acmcopyright}
\copyrightyear{2022}
\acmYear{2022}
\acmDOI{xx.xxxx/xxxxxxx.xxxxxxx}

\acmConference[KDD '22]{ACM Conference}{August 14--18, 2022}{Washington DC, USA}
\acmPrice{15.00}
\acmISBN{978-1-4503-XXXX-X/22/06}

\usepackage{algorithm}
\usepackage{algorithmic}

\usepackage{threeparttable}

\usepackage{array}
\usepackage{multirow}
\newcommand{\tabincell}[2]{\begin{tabular}{@{}#1@{}}#2\end{tabular}}  

\usepackage{subfigure}

\usepackage{amsmath, bm}
\usepackage{color}

\newtheorem{proposition}{Proposition}[section]

\begin{document}

\title{TRacer: Scalable Graph-based Transaction Tracing for Account-based Blockchain Trading Systems}
% TRacker: Transaction Tracking on Blockchain Trading Systems via personalized PageRank-based Graph Expansion

\author{Zhiying Wu}
\authornotemark[1]
\email{wuzhy95@mail2.sysu.edu.cn}
\orcid{1234-5678-9012}

\author{Jieli Liu}
\authornote{Both authors contributed equally to this research.}
\email{liujli7@mail2.sysu.edu.cn}
\affiliation{%
  \institution{Sun Yat-sen University}
  \city{Guangzhou}
  \state{Guangdong}
  \country{China}
}

\author{Jiajing Wu}
\authornote{Corresponding author.}
\email{wujiajing@mail.sysu.edu.cn}
\affiliation{%
  \institution{Sun Yat-sen University}
  \city{Guangzhou}
  \state{Guangdong}
  \country{China}
}

\author{Zibin Zheng}
\email{zhzibin@mail.sysu.edu.cn}
\affiliation{
  \institution{Sun Yat-sen University}
  \city{Guangzhou}
  \state{Guangdong}
  \country{China}
}

\renewcommand{\shortauthors}{Zhiying Wu and Jieli Liu, et al.}

\begin{abstract}

Security incidents such as scams and hacks, have become a major threat to the health of the blockchain ecosystem, causing billions of dollars in losses each year for blockchain users.
% Due to the pseudonymous nature of blockchain, it is difficult to reveal the real-world entities behind the blockchain transactions and recover the stolen funds from the massive transaction data. Recently, much effort has been devoted to tracing the flow of funds involved in illegal transactions.
To reveal the real-world entities behind the pseudonymous blockchain account and recover the stolen funds from the massive transaction data, much effort has been devoted to tracing the flow of illicit funds in blockchains recently.
However, most current tracing approaches based on heuristics and taint analysis have limitations in terms of universality, effectiveness, and efficiency. This paper models the blockchain transaction records as a blockchain transaction graph and tackles blockchain transaction tracing as a graph searching task. 
We propose \textbf{TRacer}, a scalable transaction tracing tool for account-based blockchains. 
To infer the relevance between accounts during graph searching, we develop a novel personalized PageRank method in TRacer based on the directed, weighted, temporal, and multi-relationship blockchain transaction graphs. 
To the best of our knowledge, TRacer is the first intelligent transaction tracing tool in account-based blockchains that can handle complex transaction actions in decentralized finance (DeFi).
Experimental results and theoretical analysis prove that TRacer can complete the transaction tracing task effectively at a low cost. 
All codes of TRacer are available at GitHub \footnote{https://github.com/wuzhy1ng/BlockchainSpider}.
% , whose data come from the open APIs \footnote{https://blockscan.com/} ensuring everyone to trace the illicit funds in the account-based blockchains without deployment.
\end{abstract}





\maketitle

\section{Introduction}
% 介绍区块链交易系统背景
% Since the debut of Bitcoin in 2009 \cite{nakamoto2008bitcoin}, various cryptocurrencies and blockchain technology which provide decentralized environments for cryptocurrency transactions have gained increasing popularity and attention in recent years~\cite{wu2021analysis}. 
% Incorporating peer-to-peer (P2P) technology, cryptography, and consensus protocol, blochchain is regarded as the fundamental technology supporting cryptocurrencies and allows users to participate in cryptocurrency trading with pseudonyms~\cite{zhao2021temporal}.
The rapid development of blockchain technology has aroused great attention of businesses and researchers recently. By incorporating peer-to-peer networks, cryptography, and consensus protocols, blockchain \cite{zheng2018blockchain} achieves a decentralized environment for trading and brings new vitality to traditional industries. Particularly, the typical account-based blockchain platform Ethereum has opened the era of blockchain 2.0 through the introduction of smart contracts, giving blockchain various possibilities of application. However, the pseudonymous nature of blockchain has also attracted a variety of illegal transaction activities like financial scams and hacks.
% For example, a vulnerability of ``The DAO'' smart contract in Ethereum \cite{wood2014ethereum} was exploited by hackers in 2016, and a large amount of investment worth over \$70 million was stolen \cite{chainanalysis2016}.
% According to a report given by Chainalysis \cite{chainanalysis2019}, a blockchain data analysis service provider, cryptocurrency transactions with a total amount of more than \$11.5 billion worth of cryptocurrency transactions were associated with illegal transaction activities. % 这个报告是2019的
According to a recent report of Chainalysis \cite{chainalysis2022crime}, a famous blockchain security company, the losses caused by illegal transaction activities in cryptocurrency-related businesses have exceeded \$14 billion during 2021. Along with the boom of DeFi, most of these illegal trading activities and malicious attacks are conducted in the account-based blockchain trading systems like Ethereum and Binance Smart Chain (BSC) \cite{binancesmartchain} since the boom of DeFi.
% Due to the enforcement of the know-your-customer (KYC) process in some blockchain exchanges, illegal profits on many blockchain trading systems are usually laundered into concealed and ``clean'' fund flows before being cashed out. 
% Therefore, in order to crack down on illegal transaction activities on blockchain, a wealth of efforts from both academia and industry have been devoted to tracking the flow of funds involved in illegal transaction activities \cite{yousaf2019tracing,wulei2021www,di2015bitconeview}, the purpose being to understand real-world entities behind each transaction, help victims recover the losses, and thus make it easier to counter money laundering, fraud, dark web trading, and other digital asset crimes. 

%% 第二段  描述图1异常交易检测和交易追踪的区别
With the publicly accessible blockchain transaction data, various technologies have been developed to combat financial crimes in blockchain trading systems \cite{wu2021analysis}, and these technologies can be divided into two categories, namely \textbf{proactive (pre-trade) risk warning} and \textbf{remedial (post-trade) money tracking}. Proactive risk warning refers to evaluating the risk of new transactions according to the historical behaviors of the related accounts and the existing label information. Data-based fraud detection \cite{WhoAreThePhisher}, attack detection \cite{wu2021defiranger}, and other types of illegal transaction detection technologies \cite{weber2019anti} can be categorized in this scope. However, as shown in Figure \ref{fig:traditional_fraud_detection}, although proactive early warning can raise warnings for risky transactions before the occurrence of the trade, it cannot prevent the criminals who has already gotten the money from laundering and cashing out of the ill-gotten gains from exchanges since the pseudonymous and distributed nature of blockchain systems. Therefore, a wealth of efforts have been devoted to the remedial money tracking of the ill-gotten gains \cite{yousaf2019tracing,wulei2021www,di2015bitconeview}, aiming to deanonymize the related money flows and help victims recover the losses. 

Figure \ref{fig:blockchain_transaction_tracking} shows a toy example of the remedial money tracing. 
Generally speaking, transaction tracing starts from a source and traces the flow of money to the targets. Here, the source represents a tracing object such as the blockchain account of a fraudster decamping with a large number of ill-gotten gains, and the targets indicate the accounts used to gather the ``clean'' funds awaiting cashing out.
% Usually, the targets are addresses possessing the laundered funds or exchange deposit addresses. 
Though the identity information of accounts is unknown in blockchain systems, once we locate target deposit addresses of exchanges that enforce Know-Your-Customer (KYC) processes, the related criminals can be identified and caught off-chain according to the KYC information provided by the exchanges \cite{9332279}.
\begin{figure}[t]
    \centering
    \subfigure[Proactive risk warning in blockchain]{
        \includegraphics[width=0.8\linewidth]{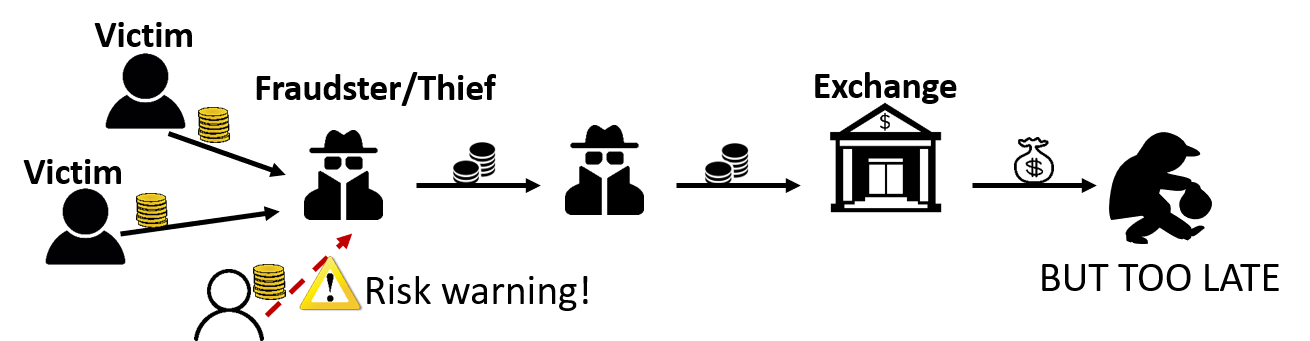}
        \label{fig:traditional_fraud_detection}
    }
    
    \subfigure[Remedial money tracking in blockchain]{
        \includegraphics[width=0.8\linewidth]{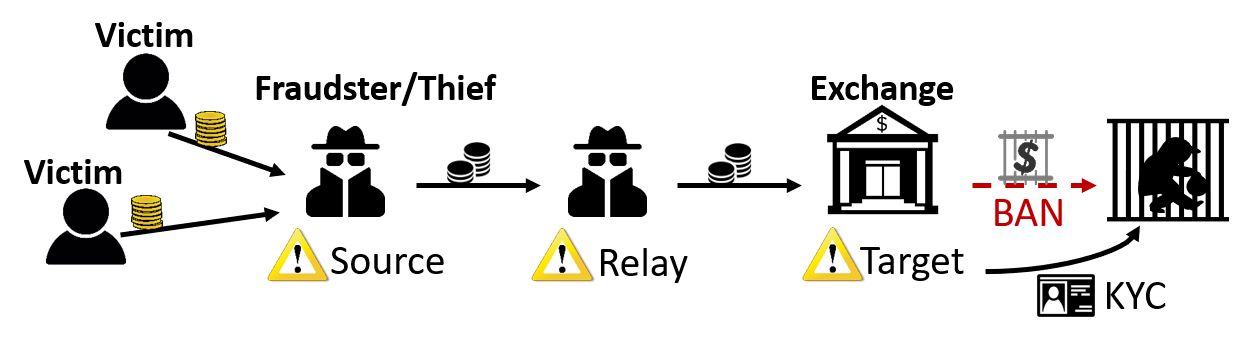}
        \label{fig:blockchain_transaction_tracking}
    }
    \vspace{-0.1in}
    \caption{(a) Proactive risk warning in blockchain raises warnings for risky transactions. %However, it is always too late to prevent criminals from laundering and cashing out of the ill-gotten gains from exchanges. 
    (b) Remedial money tracking in blockchain traces the money flow among the source and the targets possessing the laundered funds, also providing evidence to catch the criminals off-chain.}
    % 不足以阻止已发生的损失
    \vspace{-0.1in}

\end{figure}

% 区块链交易系统上开展交易追踪的挑战
% Yet transaction tracking on blockchain trading systems is a rather challenging task.
% \textbf{\textit{Firstly}}, because of the pseudonymous nature of blockchain, it is unlikely to enforce the KYC process to verify the identities and ascertain the potential risks of users during cryptocurrency transactions \cite{wu2021analysis}, which makes it extremely difficult to determine the flow direction of a certain amount of money.
% \textbf{\textit{Secondly}}, without identity information, conducting money tracking in a large amount of blockchain transaction records is like finding a needle in the haystack, which requires a low computational cost of the transaction tracking.

% 当前工作的问题
% Current approaches for blockchain transaction tracking \cite{zhao2015graph,phetsouvanh2018egret,oggier2020ego} are mainly based on the Bitcoin system and inspired by graph searching \cite{xu2004fighting,abiteboul2003adaptive} and taint analysis \cite{moser2014towards,tironsakkul2019probing} technologies.
% The main idea of existing studies is to start from the source node and search the possible paths of funds between the source and target nodes via certain rules.
Current approaches for blockchain transaction tracing \cite{zhao2015graph,phetsouvanh2018egret,oggier2020ego,tironsakkul2019probing} are mainly based on rule-based heuristics and taint analysis~\cite{moser2014towards}. 
However, as an emerging research topic, existing transaction tracing methods have limitations in terms of universality, effectiveness and efficiency. Particularly, most existing heuristic methods are designed for specific scenarios based on experts experience, and cannot be automatically and intelligently applied to various blockchain transaction scenarios. In addition, the time cost and end conditions of existing methods, especially those requiring manual verification and intervention, is not definite or well defined, making their time efficiency and effectiveness difficult to guarantee.
% 然而，大多数现有方法通常依赖于专家经验来针对具体案例进行设计，因此很难自动且智能地适用于各类区块链交易场景（普适性）。此外，现有方法的时间复杂性较高，且没有确定的和well designed end condition，使用它们在巨大的交易网络中时犹如大海捞针一样时间效率和有效性都难以保障。
% However, as a newly emerging area, there still exist some limitations in much of the existing work:
% \textbf{(1)} Without theoretical proofs, the end condition of most of the existing methods relies on experts, which makes it difficult to estimate the computational cost of these methods.
% Therefore, the time cost of running existing transaction tracking methods on a large-scale transaction network may be extremely high when the end condition is not well designed.
% \textbf{(2)} Most of the existing studies are heuristic methods designed for some specific cases, which rely heavily on expert experience. Currently, there is no common framework and evaluation criteria for transaction tracking on blockchain. Therefore, it is difficult to evaluate the universality and superiority of the existing methods.
%However, most of the existing methods are designed for Bitcoin and some specific scenarios like cross-chain transactions \cite{yousaf2019tracing} and mixing transactions \cite{beres2020blockchain,wulei2021www}. 
%And they are inapplicable to deal with the transaction tracing task in account-based blockchain trading systems due to the different transaction models and the dependence on expert experience. 
%Moreover, due to the massive transaction data in account-based blockchains, tracing with the manual analysis method like finding a needle in the haystack. And it is difficult to quickly respond to security incidents. 
Besides, the popularity of DeFi in account-based blockchains like Ethereum and BSC brings many new kinds of semantics to blockchain transaction actions, leading to high barriers to the transaction tracing task.
% Yet transaction tracing in blockchain trading systems is a rather challenging task like finding a needle in the haystack due to the large amount of transaction data.

\begin{figure*}[t]
    \centering
    \includegraphics[width=0.65\linewidth]{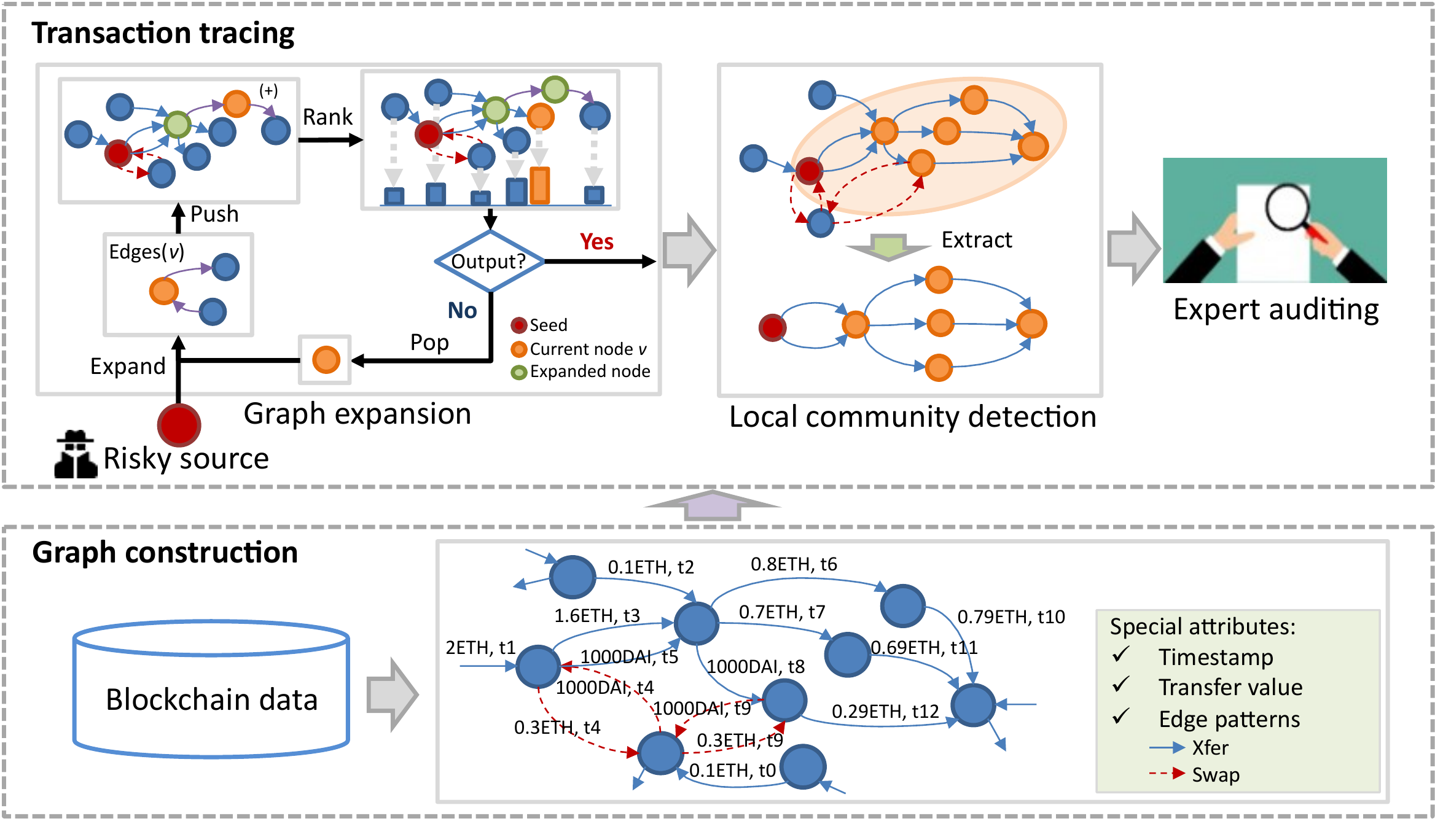}
    \caption{The framework of TRacer.}
    \label{fig:TRacer}
\end{figure*}

% 介绍本文的思路
In this paper, we propose the first general tool called \textbf{\underline{T}\underline{R}acer}, which is a transaction \underline{T}racing in account-based blockchain trading systems incorporating Personalized Page\underline{R}ank-based technologies. Referring to the framework displayed in Figure \ref{fig:TRacer}, TRacer first models the blockchain transaction data including the complex DeFi operation actions into a directed, weighted, temporal, and multi-relationship graph, and then conducts local community detection on a subgraph around the risky seed obtained by graph expansion. The output small-scale community is the traced money flow graph and can be further audited by experts.
Moreover, we propose a novel ranking algorithm based on personalized PageRank \cite{page1999pagerank} to reveal the relevance between the source and other accounts in blockchain systems. We introduce approximate personalized PageRank (APPR) \cite{andersen2006local,andersen2007local,yin2017local} to obtain the approximate solution of personalized PageRank, which can improve the scalability of the algorithm. Both graph expansion and local community detection in TRacer are based on the proposed ranking algorithm. 
%For the problems of effectiveness quantification, we design some metrics for understanding what kind of transaction tracking methods are effective. 
% It is worth mentioning that our proposed method has been applied and evaluated in transaction tracking tasks on several blockchain trading systems including Bitcoin, Ethereum, and Binance Smart Chain \cite{binancesmartchain} and so on, but may not be suitable for some privacy-enhancing blockchain trading systems like the Monero \cite{noether2014monero}. % 这些匿名系统都是UTXO上的
% Finally, we evaluate our method with metrics and network visualization analysis on five Ethereum transaction tracing cases.
Experimental results demonstrate the performance of TRacer on transaction tracing in account-based blockchains. The main contributions of this work can be summarized as follows:
\begin{itemize}
    % \item \textbf{Problem}: We give a problem definition of transaction tracking on Blockchain trading systems as finding the paths among the source node and the target nodes, and propose a general transaction tracking model named Push-Pop model.
    \item To the best of our knowledge, we are the first to study intelligent transaction tracing in account-based blockchain trading systems like Ethereum, which is an urgent problem with the booming security incidents in these systems.
    % \item \textbf{Approach}: We propose the TTR algorithm running on a transaction network, whose cost for convergence is $O(\frac{1}{\epsilon \alpha})$ controlled by two constant parameters of $\epsilon$ and $\alpha$. Additionally, we give the theoretical proof on how to set the parameters for ensuring our method is able to track the paths among the source node and the target nodes, even in the face of the worst conditions.
    \item We design and implement a general blockchain transaction tracing tool named TRacer, which is able to incorporate the complex semantics of transaction actions in DeFi. Particularly, a novel personalized PageRank method is employed to estimate the relevance score of accounts in TRacer.
    % \item \textbf{Evaluation}: We introduce metrics to quantify the effectiveness of transaction tracking methods on Blockchain trading systems, and experimental results demonstrate the priority of the proposed method. Moreover, we collect five standard datasets in Ethereum from real cases for blockchain transaction tracking and anti-money laundering research, and these codes for experiments were published on our Github\footnote{https://github.com/wuzhy1ng/BlockchainSpider}.
    % We thoroughly evaluate the proposed approach by comparing with the existing benchmarks on the historical loan data and achieved state-of-the-art performance. In addition, we conduct empirical studies in real-world risk control applications, and the result proves our method could prevent major financial losses for the financial institution
    \item We thoroughly evaluate the performance of TRacer via both theoretical analysis and experimental evaluation, and the results demonstrate the effectiveness and the scalability of TRacer. We also contribute a benchmark dataset verified by several security companies for evaluating the transaction tracing methods.
\end{itemize}

\section{Related Work}
To combat crimes in blockchain, transaction tracing is a critical task for blockchain security companies and the community. Current methods for blockchain transaction tracing are mainly designed for Bitcoin-like blockchain systems and heavily rely on expert experience. The most widely used transaction tracing method is Breath First Search (BFS) \cite{zhao2015graph,phetsouvanh2018egret}, and many related tools like skytrace\footnote{https://www.certik.com/skytrace} and coinholmes\footnote{https://trace.coinholmes.com} are available. Since BFS is insufficient to reveal the audit priority of different money tracking directions, the taint analysis technologies \cite{moser2014towards,tironsakkul2019probing,di2015bitconeview} have been proposed for bitcoin tracing according to the amount value and the order of multiple outputs in each transaction. However, these technologies rely on expert analysis and cannot automatically output the money flows from the source to the targets.

Based on the co-spending clustering heuristic in Bitcoin that the input addresses of the same transaction belong to the same entity \cite{Reid2013}, Huang et al. proposed to track the bitcoin flow of ransomware to when the money is being cashed out \cite{HuangTracking2018}.
% Tracking Ransomware End-to-end
Meiklejohn et al. \cite{meiklejohn2013fistful} utilized the change address heuristic to deanonymize the money flows in Bitcoin.
Moreover, some rule-based transaction tracing methods were proposed for cross-chain scenarios and mixing scenarios. For example, Yousaf et al. traced the cross-chain money flows by identifying the transactions which happen close in time and have similar amount value \cite{yousaf2019tracing}. 
However, these rule-based methods are designed for specific protocols and they are inapplicable to transaction tracing in account-based blockchains.

Personalized PageRank models the relevance of nodes in a network to a specific node, and has been widely used in web search \cite{page1999pagerank}, recommendations \cite{WTFGupta2013}, etc. Recently many variants and efficient approximations of personalized PageRank have been proposed for large-scale applications \cite{bojchevski2020scaling}. In this paper, we model the account-based blockchain data as a directed, weighted, temporal, and multi-relationship graph, and formalize the transaction tracing problem as a graph searching problem. We develop a scalable and intelligent transaction tracing tool for account-based blockchains based on personalized PageRank and its approximate solutions.
\section{Preliminaries} 
\label{sec:preliminaries}
% Before introducing our method, this section discusses the problem definition and provides preliminaries for this paper.

\subsection{Account-based blockchains}
% 不同于比特币系统采用基于UTXO的交易模型，以以太坊为首的基于账户的区块链系统xxxx
% 智能合约  DApp
% 由DApp引出Token  Token标准
% 去中心化交易所
Traditional Bitcoin-like blockchains are based on the transaction-centered model \cite{wu2021analysis} whose building block is unspent transaction output (UTXO), which is an indivisible cryptocurrency chunk locked to a specific owner. Each transaction has multiple inputs made up of UTXOs and multiple outputs, and there is a change address in the outputs used to receive the change.

Different from UTXO-based blockchains, account-based blockchain systems like Ethereum and BSC have the concept of account similar to that of traditional banking accounts. 
There are two kinds of accounts in account-based blockchains, namely external owned account (EOA) and smart contract account. 
Accounts are the initiators of blockchain transactions and record some dynamic state information including account balance. 
Especially, each smart contract account is associated with a piece of executable bytecode. There are also two types of transactions in the systems. 
A transaction triggered by an EOA is called external transaction, while a transaction triggered by an invocation of the function in a smart contract account is called internal transaction. 
In addition, an EOA can invoke functions of a smart contract in an external transaction and further result in many internal transactions. Transaction hash consisting of a set of numbers and letters is used to uniquely identify a particular transaction from or to an EOA.
Due to the support of smart contract, everyone can take the advantage of blockchain technology and build DApp projects in account-based blockchains.
Besides the native currency in blockchain systems, there are many third-party tokens representing assets, currency, or access rights of projects in the account-based blockchain ecosystem. 
To facilitate token development and exchange, some token standards are launched in blockchain trading systems, e.g., the ERC20 token standard in Ethereum. 
There are also many DeFi DApps that offer financial services such as token lending and exchange.
% 原来：
% 介绍加密数字货币包括原生币和Token
% There are two types of cryptocurrencies in account-based blockchains: native token and third-party token complimenting some standards.
% For Ethereum, the native token Ether and 
% 介绍什么是DeFi，DeFi世界中的交易语义有哪些

\subsection{Problem Definition}
\label{sec:problem_define}
The transaction tracing task in account-based blockchain trading systems aims to trace the money flows from a given source to the targets that gather the money awaiting cashing out and point the priority money flows for auditors to further verify manually. By modeling the blockchain transaction relationships as graph where nodes indicate the accounts and edges indicate the money flow relationships, we can formulate the transaction tracing problem as follows. 

% \begin{prbdef}
\textbf{Problem formulation.} \textit{Given a source node $s$ in a transaction graph $G$, the goal is to search a connected money transfer subgraph $G_{s}$ from $s$ to its money flow targets around the neighborhood of $s$. $G_{s}$ should contain as many target nodes as possible in the smallest possible size of graph for manual verification.}

\subsection{Approximate Personalized PageRank}
% To ensure the effectiveness of the tracking model, we usually need to define an importance measurement of other nodes in the network to the source node.
% or the probability of random walk to $u$ starting from $s$
\textbf{Personalized PageRank.} The personalized PageRank vector $\bm{p}_s$ of a source node $s$ in a graph $G=(V,E)$ is defined as the unique solution of Equation \ref{equ:ppr}, i.e.,
\begin{equation}
    \bm{p}_s = \alpha \bm{e}_s + (1-\alpha) \bm{p}_s M,
    \label{equ:ppr}
\end{equation}
where $\alpha$ is a teleport constant between 0 and 1, $\bm{e}_s$ is the indicator vector with a single nonzero element of 1 at $s$, $M$ is a transition matrix and $M=D^{-1}A$ where $D$ and $A$ are degree matrix and adjacency matrix. The definition of personalized PageRank is equivalent to simulating a random walk starting from $s$, and $\bm{p}_s$ is a probability vector where $\bm{p}_s(u)$ is the probability that a certain random walk beginning at $s$ terminates at $u$. 
% Let $\bm{p}_s$ denote the personalized PageRank vector of a source node $s$ in a graph $G=(V,E)$, where $\bm{p}_s(u)$ describes the relevance of a node $u$ to the source node $s$.

\textbf{Approximate personalized PageRank (APPR).} The first algorithm proposed to calculate personalized PageRank is power iteration \cite{page1999pagerank}, which requires high time complexity and not effective in large-scale networks.
Therefore, various efficient approximate solutions of personalized PageRank have been proposed, and the most widely used one is called the ``local push'' algorithm \cite{andersen2006local}. This algorithm starts with all probability residual on the source node of the graph, and pushes the residual to the neighbors iteratively. During the iterations, the residual of each node can be transformed into the relevance to the source node. Finally an estimate of $\bm{p}_s$ can be obtained. In the algorithm, a residual vector $\bm{r}_s$ is used to maintain the residual, where $\bm{r}_s(u)$ denotes the residual of node $u$. By setting $\bm{p}_s=\vec{0}$ and $\bm{r}_s=\bm{e}_s$ for initialization, the local push procedure updates the value of $\bm{p}_s(u)$ as follows:
\begin{equation}
    \begin{cases}
    & \bm{p}_s(u)=\bm{p}_s(u)+\alpha \bm{r}_s(u) \\
    & r(v)=r(v)+(1-\alpha)r(u) / d(u), \\
    % & \bm{r}_s(v)=\bm{r}_s(u)+\theta(u,v)\bm{r}_s(u)
    \end{cases},
    \label{equ:local_push_appr}
\end{equation}
where $v \in N(u)$ is the neighbor of $u$, and 
% the push factor $\theta(u,v)$ between $u$ and $v$ is defined as
% \begin{equation}
%     \theta(u,v)=\frac{1-\alpha}{d(u)},
% \end{equation}
% where 
$d(u)$ is the degree of $u$.
The local push procedure stops when the residual of each node in $G$ is within $\epsilon$.

\section{Proposed Approach}
\label{sec:proposed_approach}
% This section introduces the Push-Pop model and designs the TTR algorithm in Push-Pop model for blockchain transaction tracking.
This section introduces the architecture, implementation, and theoretical analysis of TRacer.

\subsection{Architecture Overview}
TRacer consists of three main modules including graph construction, graph expansion, and local community detection. Each module is described as follows:
\begin{itemize}
\item \textbf{Graph construction}: This module models the money transfer relationships between accounts as directed, weighted, temporal, and multi-relationship graphs.
% in which the nodes and the edges represent the accounts and the transactions in the blockchain, respectively.
\item \textbf{Graph expansion}: Since the blockchain data contains billions of transactions which is too large for common graph algorithms, this module aims to find a relevant subgraph from a risky source. The module contains four operations: \textit{\textbf{Expand}} collects all edges related to a given node, \textit{\textbf{Push}} merges the collected edges to the subgraph, \textit{\textbf{Rank}} computes the relevance of nodes in the subgraph to the source, and \textit{\textbf{Pop}} selects a node for expanding.
The graph expansion process terminates when the end condition is met.
\item \textbf{Local community detection}: This module \textit{\textbf{Extract}}s a local community of the source node from the expended subgraph, in which nodes have higher relevance to the source than nodes out of the community.
\end{itemize}

\subsection{Graph construction}
Since there are multiple types of tokens in account-based blockchain trading systems, we formulate the money transfer relationships among accounts into a directed, weighted, temporal, and multi-relationship graph $G=(V, E)$, where $V$ is the node set representing accounts, $E$ is the edge set representing the token transfer relationships.
% with $\varphi_E: E\rightarrow\{{\rm token \ types}\}$ for edge-type mapping
% , and $X$ denotes the set of attributes attached to edges including transaction hash (unique to the corresponding external transaction), amount, timestamp, and token symbol. 
An edge $e=(u,v,w,t,b,h)$ denotes that account $u$ transfer $w$ units token $b$ to account $v$ at timestamp $t$ with a transaction hash $h$. We define mapping functions $f_{src}$, $f_{tgt}$, $f_{amt}$, $f_{ts}$, and $f_{sym}$ to map each edge to its source, target, amount, timestamp, and token type respectively. There are multiple types of edges indicating the transfer relationships of different tokens.

In addition, many blockchain services such as decentralized exchanges act as the intermediary for token swap. 
To reveal the token flows after users interact with these services, we categorize the money transfer relationships into two patterns, namely \textbf{Xfer} and \textbf{Swap}.
As shown in Figure \ref{fig:Xfer}, accounts send or receive tokens through the Xfer pattern, and the related DeFi actions \cite{wu2021defiranger} in this pattern include: 1) \textbf{\textit{transfer}}: An account sends an amount of token to another account, 2) \textbf{\textit{minting}}: A token contract mints an amount of token to an account, and 3) \textbf{\textit{burning}}: An account burns an amount of token.
While accounts exchange tokens for other kinds of tokens through the Swap pattern as shown in Figure \ref{fig:Swap}, including three related DeFi actions: 1) \textbf{\textit{add liquidity}}: An account deposits an amount of token to a DeFi app, and receives a certain amount of Liquidity Provider (LP) token back, 2) \textbf{\textit{remove liquidity}}: An account sends an amount of LP token to a DeFi app for destroying and gets a certain amount of other tokens back, and 3) \textbf{\textit{trade}}: An account sells an amount of token A in a DeFi app for a certain amount of token B. We identify the transaction relationships involving both the sending of tokens and the receiving of another kind of tokens with the same hash as the relationships in the Swap pattern, and otherwise as the relationships in the Xfer pattern. 
% In this way, we define the pattern mapping for the edges in the transaction graph:
% \begin{equation}
%     P: E \rightarrow \{Xfer, Swap\}.
% \end{equation}

% In addition, considering the token transfer, there are two types of patterns for a transaction in TRacer, namely \textbf{Xfer} and \textbf{Swap}.
% A transaction is recognized as Swap if the transactions with the same transaction hash ensure an account sending and receiving tokens at the same time, otherwise Xfer.
% These two patterns are summarized from the transaction actions defined by token contracts and DeFi DApps. 
% Several main actions \cite{wu2021defiranger} are listed below:
% \begin{itemize}
%     \item \textbf{Transfer}: An account sends an amount of token to another account.
%     \item \textbf{Minting}: A token contract sends an amount of minted token to an account. 
%     \item \textbf{Burning}: An account sends an amount of token to a token contract for destroying.
%     \item \textbf{Add Liquidity}: An account deposits an amount of token to the DeFi app, which also sends its an amount of Liquidity Provider (LP) token to the account.
%     \item \textbf{Remove Liquidity}: An account sends an amount of LP token to a DeFi app for destroying, and the account gets a number of other tokens back.
%     \item \textbf{Trade}: An account sells an amount of token A in a DeFi app for a certain amount of token B.
% \end{itemize}

\begin{figure}[t]
    \centering
    \subfigure[The Xfer pattern]{
    \includegraphics[width=0.3\linewidth]{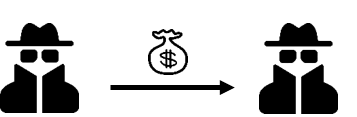}
    \label{fig:Xfer}
    }
    \subfigure[The Swap pattern] {
    \includegraphics[width=0.35\linewidth]{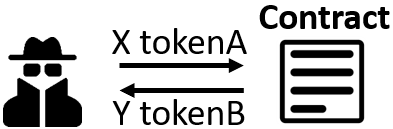}
    \label{fig:Swap}
    }
    \caption{(a) Xfer: Sending or receiving tokens.
    % Transaction actions in this pattern include transfer, minting, and destroying. 
    (b) Swap: Exchanging tokens for other kinds of tokens. 
    % Transaction actions in this pattern include add liquidity, remove liquidity, and trade.
    }
\end{figure}

\subsection{Graph Expansion}
The graph expansion module aims to obtain a relevant subgraph expanded from a given seed via iterating four operations shown in Figure \ref{fig:TRacer}. In what follows, we introduce the design of these four operations.
\subsubsection{\textit{Push} and \textit{Rank}: Transaction Tracing Rank}
The \textit{Push} operation merges the expanded edges to the subgraph in each iteration. While \textit{Rank} calculates the relevance of nodes in the subgraph to the source.
% While Rank is the core operation of graph extension in TRace, which determines the rules for calculating the relevance of nodes to the source node in the subgraph.

We introduce APPR to calculate the node relevance. 
More specifically, the Rank operation executes the local push procedure for the incremental update of node relevance.
% Especially, according to some characteristics of the transaction graph in account-based blockchain, we design a novel local push procedure. 
% The relevance computed using this novel local push procedure is named \textbf{Transaction Tracing Rank} (TTR).
Based on the characteristics of blockchain transaction graph in account-based blockchains, we propose a novel local push algorithm named \textbf{\underline{T}ransaction \underline{T}racing \underline{R}ank (TTR)} to compute the node relevance in blockchain transaction graph. We develop four local push strategies in TTR, namely \textbf{tracing tendency, weighted pollution, temporal reasoning, and token redirection}. We use $r_s(u,t,b)$ to denote the residual of node $u$ brought from token $b \in B$ at timestamp $t \in \Gamma$, where $\Gamma$ is the timestamp set and $B$ is the token symbol set in the subgraph.
% In the local push procedure of TTR, the edge attributes of time and token symbol need to be considered, thus we use $r_s(u,t,b)$ to denote the residual attached with a timestamp $t \in \Gamma$ and a token symbol $b \in B$ in the node $u$, where $\Gamma$ is a set of timestamps, and $B$ is a set of token symbols.
% % the residual is stored by a dictionary structure with the definition below:
% % \begin{equation}
% %     r_s: V \times \Gamma \times B \rightarrow \mathbb{R},
% %     \label{req:residual_mapping}
% % \end{equation}
% % where $V$ denotes a set of nodes, $\Gamma$ denotes a set of timestamps, and $B$ denotes a set of token symbols.
In this way, the local push procedure in TTR is described as Algorithm \ref{alg:ttr_local_push}. During initialization, a part of the residual of node $u$ is transformed into the relevance (line 1). Moreover, other parts of the residual are pushed to the neighbors according to the four strategies.

\begin{algorithm}[t]
    \caption{TTR local push}
    \label{alg:ttr_local_push}
    \begin{algorithmic}[1]
        \REQUIRE Node $u$, edges $E(u)$ related to $u$, edge mapping functions $f_{src}$, $f_{tgt}$, $f_{amt}$, $f_{ts}$, $f_{sym}$, rank $\bm{p}_s$, residual $r_s$,
        timestamps $\Gamma$, token symbols $B$, 
        teleport constant $\alpha$, and tracing tendency $\beta$.
        \ENSURE The updated $\bm{p}_s$ and $r_s$.
        %% 自环
        % \STATE $\Gamma=\{f_{ts}(e) | e\in E(u)\}$
        % \STATE $B=\{f_{sym}(e) | e\in E(u)\}$
        \STATE $\bm{p}_s(u) = \bm{p}_s(u) + \alpha \sum\limits_{b \in B}\sum\limits_{t \in \Gamma}r_s(u,t,b)$
        \STATE $r'_s(u,t,b)=r_s(u,t,b), r_s(u,t,b)=0$, $\forall t\in \Gamma, \forall b\in B$.
        \FOR{$(t,b) \in \Gamma \times B$}
            % \STATE $res = r'_s(u,t,b)$
            %% 出边
            \STATE $E_{out}=\{e\in E(u)|f_{ts}(e)>t\wedge f_{src}(e)=u\wedge f_{sym}(e)=b\}$
            \STATE $E_{in}=\{e\in E(u)|f_{ts}(e)<t\wedge f_{tgt}(e)=u\wedge f_{sym}(e)=b\}$
            \FOR{$E' \in \{ E_{out}, E_{in} \}$}
                % \STATE $W=\sum\limits_{e \in E'}f_{amt}(e)$.
                \STATE $\gamma = \beta \ if \  E'==E_{out} \ else \ (1-\beta)$
                \FOR{$e \in E'$}
                    \STATE $E'_{\rho}=\rho(\{e\}, E(u))$ // obtained by Equation \ref{equation:token_redirect}
                    \FOR{$e' \in E'_{\rho}$}
                        \STATE $\Delta = \frac{(1-\alpha)\gamma f_{amt}(e)}{|E'_{\rho}|\sum\limits_{e ''\in E'}f_{amt}(e'')} r'_s(u,t,b)$
                        \STATE $v=f_{src}(e') \ if \ E'==E_{in} \ else \ f_{tgt}(e')$
                        \STATE $r_s(v,f_{ts}(e'),f_{sym}(e'))+=\Delta$
                    \ENDFOR
                    % \STATE $r_s(f_{amt}(e'),f_{ts}(e'),f_{sym}(e'))+=\frac{(1-\alpha)\gamma f_{amt}(e)}{|E'_{\rho}|\sum\limits_{e \in E'}f_{amt}(e)} r'_s(u,t,b)$ for $e'\in E'_{\rho}$
                \ENDFOR
                \STATE $r_s(u,t,b)=(1-\alpha)\gamma r'_s(u,t,b)$ if $|E'|==0$.
            \ENDFOR
            % \STATE $r_s(u,t,b)=(1-\alpha)\beta r'_s(u,t,b)$ if $|E_{out}(u)|==0$.
            % \STATE $r_s(u,t,b)=(1-\alpha)(1-\beta) r'_s(u,t,b)$ if $|E_{in}(u)|==0$.
        \ENDFOR
        \RETURN $\bm{p}_s, r_s$
    \end{algorithmic}
\end{algorithm}
% In this way, the TTR local push procedure in a node $u$ can be described by the equations below:
% \begin{equation}
%     \begin{cases}
%         & \Delta \bm{p}_s(u) = \alpha \sum\limits_{b \in B} \sum\limits_{t \in \Gamma} r_s(u,t,b) \\ 
%         & \Delta r_s(v_{out},t,b) = (1-\alpha)\beta \rho 
%         \begin{pmatrix}
%         b, \theta 
%             \begin{pmatrix}
%             v_{out}, r_s, \tau(t,E(u))
%             \end{pmatrix}
%         \end{pmatrix}
%          \\ 
%         & \Delta r_s(v_{in},t,b) = (1-\alpha)(1-\beta) \rho
%         \begin{pmatrix}
%         b, \theta 
%             \begin{pmatrix}
%             v_{in}, r_s, \tau(t,E(u))
%             \end{pmatrix}
%         \end{pmatrix}
%     \end{cases},
% \end{equation}
% in which $t\in \Gamma$, $b \in B$, $v_{out} \in N_{out}(u)$ and $v_{in} \in N_{in}(u)$ denote out-degree neighbor and in-degree neighbor of $u$ respectively, $\beta$ denotes the tracing tendency coefficient controlling the attention for different neighbors, $\theta$ denotes a weight pollution function focusing on distributing the residual by edge weight, $\tau$ denotes the temporal reasoning function aiming to select edges with specific timestamps, $\rho$ denotes the token redirection function for processing the transaction patterns, and $E(u)$ denotes the edges related to the node $u$.

\textbf{Tracing tendency}.
Since a money transfer relationship between accounts is directed, the attention to in-degree neighbors and out-degree neighbors during the local push procedure can be different.
% For example, tracing for the target of the fund flows oriented from a particular source needs to pay more attention to its out-degree neighbors, while tracing for the source of money needs to pay more attention to the in-degree neighbors.
In most cases, tracing for the destinations of money flows oriented from a particular source needs to pay more attention to the out-degree neighbors. Therefore, we define an attention coefficient $0 \leq \beta \leq 1$. During the residual propagation, the out-degree neighbors in a transaction relationship can get a higher residual when $\beta > 0.5$, and the in-degree neighbors can receive a higher residual when $\beta < 0.5$. As Figure \ref{fig:TTR_strategies}(a) shows, the in-degree neighbors and the out-degree neighbors obtain a different weight in this strategy. Line 6-17 in Algorithm \ref{alg:ttr_local_push} describe the residual propagation procedure of different directions with the attention coefficient assigned at line 7. 
%Therefore, considering the transaction graph $G$ is directed and giving a tracing tendency coefficient $0 \leq \beta \leq 1$, the out-degree neighbors in a transaction relationship get a higher residual when $\beta > 0.5$ for tracing the target of the fund flow, and the in-degree neighbors in a transaction relationship get a higher residual when $\beta < 0.5$ for tracing the source of the fund flow.
% As Figure \ref{fig:tracing_tendency} shows, the out-degree neighbors of a node $u$ get a higher residual than the in-degree neighbors, when $\beta > 0.5$.
% In this way, the residual is split for pushing to out-degree neighbors and in-degree neighbors with different coefficients being $\beta$ and $1-\beta$ (line 7) in Algorithm \ref{alg:ttr_local_push}.

\textbf{Weighted pollution}. 
% Traditional APPR pushes the residual of a node to the neighbors by considering an equal weight during a local push procedure. 
% However, in a transaction relationship, the transaction strength is usually weighted by the transaction amount. 
% Note that the edge strength is usually weighted by the transaction amount.
In blockchain transaction graph, the strength of money transfer relationships is weighted by the transaction amount. For each node, a neighbor trading a larger amount of money with this node is considered to be more relevant to the node. 
% Therefore, by considering the transaction graph $G$ is weighted with the transaction amount information as the weight, 
Therefore, in this strategy we take the weight of the money transfer relationships into account. Figure \ref{fig:TTR_strategies}(b) shows an example of this strategy that neighbors of $u$ associated by edges with greater weight can obtain more residual during the propagation.
In this way, the residual of a node is distributed to each neighbor associated by edge $e \in E'$ according to the ratio of edge weight $f_{amt}(e) / \sum_{e'' \in E'}f_{amt}(e'')$ as shown in Algorithm \ref{alg:ttr_local_push} line 11, where $E'$ is the set of edges.

\textbf{Temporal reasoning}.
Blockchain transactions are recorded in blocks chronologically. Each block contains a specific timestamp. 
% Since the fund flow transferring process is dynamic, the transaction time information can be taken into consideration in the local push procedure.
With the timestamp information, tracing for the targets of money flows usually follows the paths with increasing timestamps, while tracing back to the source of money flows follows the paths with decreasing timestamps.
Therefore, in this strategy we take the temporal information into consideration. For example in Figure \ref{fig:TTR_strategies}(c), the residual from the input transaction at $t2$ is pushed out through the outgoing edges after $t2$ and the incoming edge before $t2$. Algorithm \ref{alg:ttr_local_push} line 4-5 show the edge set for residual propagation that satisfies this strategy.
% The timestamp is used to restrict the edges for pushing (line 4, 5) in Algorithm \ref{alg:ttr_local_push}.
Moreover, the residual is pushed to the node itself if the edge set is empty in Algorithm \ref{alg:ttr_local_push} line 16, indicating that the funds have not been transferred from the node.

\begin{figure}[t]
    \centering
    \includegraphics[width=0.9\linewidth]{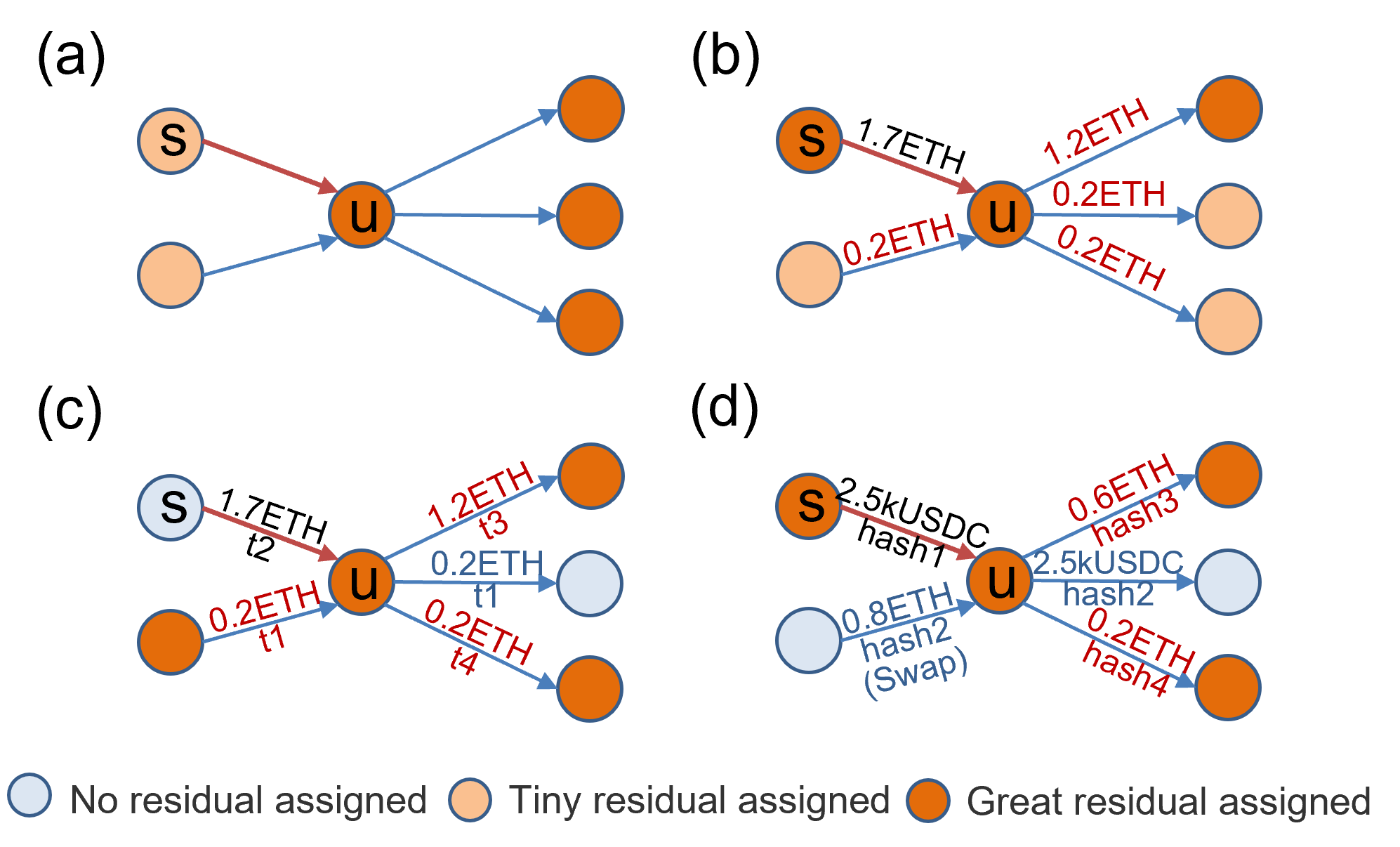}
    \caption{(a) Tracing tendency. 
    Different attention is assigned to the out-degree neighbors and in-degree neighbors.
    (b) Weighted pollution. The higher the edge weight, the closer the relationship. 
    (c) Temporal reasoning. Tracing money flow in chronological order, and $t1<t2<t3<t4$ in this example. 
    (d) Token redirection. Uncovering the token flows even though there exist complex DeFi actions in the Swap pattern.}
    \label{fig:TTR_strategies}
\end{figure}

% \begin{figure}[t]
%     \centering
%     \subfigure[Tracing tendency]{
%     \includegraphics[width=0.35\linewidth]{figures/tracing_tendency.png}
%     \label{fig:tracing_tendency}
%     }
%     \subfigure[Weight pollution]{
%     \includegraphics[width=0.35\linewidth]{figures/weight_pollution.png}
%     \label{fig:weight_pollution}
%     }
%     \subfigure[Temporal reasoning]{
%     \includegraphics[width=0.35\linewidth]{figures/temporal_reasoning.png}
%     \label{fig:temporal_reasoning}
%     }
%     \subfigure[Token redirection]{
%     \includegraphics[width=0.35\linewidth]{figures/token_redirection.png}
%     \label{fig:token_redirection}
%     }
%     \caption{
%     (a) Tracing tendency. Different attention is assigned to the out-degree neighbors and in-degree neighbors.
%     (b) Weight pollution. The higher the edge weight, the closer the relationship. 
%     (c) Temporal reasoning. Tracing money flow in chronological order, and $t1<t2<t3<t4$ in this example. 
%     (d) Token redirection. Uncovering the token flow even though there exist complex transaction actions in the Swap pattern.
%     }
% \end{figure}

\textbf{Token redirection}.
% Besides native currency transfers, the transactions of account-based blockchain contain a large number of token transfers, while token redirection makes effort to reveal the real flow of interesting tokens.
This strategy makes the effort to uncover the flow of interesting tokens based on the transaction patterns of Xfer and Swap.
% In TRacer, the transaction patterns of Xfer and Swap are taken into consideration when it comes to tracing the token fund flow.
As Figure \ref{fig:TTR_strategies}(d) shows, the residual of node $u$ brought from the USDC token in an Xfer edge with hash1 should be pushed through the edges with hash3 and hash4, rather than the USDC outgoing edge with hash2. Since the edges with hash2 are in a Swap pattern and swapped the USDC token into ETH.
% As Figure \ref{fig:token_redirection} shows, the residual with USDC token received by the node $u$ in an Xfer edge with hash1 should have been pushed to the out-degree neighbor by the edge with USDC and hash2, but considering the edge with USDC and hash2 for redirection, the residual is pushed through the output edges with hash3 and hash4.
To achieve the redirection of token flows in complex transaction actions, we define a recursive function $\rho(\cdot,\cdot)$ which can find the initial state of a set of incoming edges before the Swap operations and the final state of a set of outgoing edges after the Swap operations within a node. For example, the initial state before Swap of the incoming edge with hash2 in Figure \ref{fig:TTR_strategies}(d) is the incoming edge with hash1, and the final state after Swap of the outgoing edge with hash2 is the outgoing edges with hash3 and hash4. More specifically, as shown in Algorithm \ref{alg:ttr_local_push} line 8-14, the residual of a specific token should be pushed through the redirected edges. Therefore, $\rho(\cdot,\cdot)$ satisfies the following recursive equation:
\begin{equation}\label{equation:token_redirect}
    \rho(\mathcal{E}, E(u))=\mathcal{E}_{xfer} \cup \rho(\bigcup\limits_{e \in \mathcal{E}_{swap}} redirect(e), E(u)),    
\end{equation}
where $\mathcal{E}$ is a set of edges for redirection, $\mathcal{E}_{xfer} \subset \mathcal{E}$ is a set of Xfer edges, $\mathcal{E}_{swap} \subset \mathcal{E}$ is a set of Swap edges, and $redirect(\cdot)$ selects the edges in the token types before/after Swap for incoming/outgoing edges $e \in \mathcal{E}_{swap}$ from edges $E(u)$ related to $u$.
% , which satisfies $f_{src}(e')=f_{src}(e)=u$ or $f_{tgt}(e')=f_{tgt}(e)=u$ for $\forall e' \in redirect(e)$.
% redirect selects the output/input edges from $E(u)$ with the swapped token symbol of the output/input edge $e \in \mathcal{E}_{swap}$.
Especially, $\rho(\emptyset, E(u)) = \emptyset$.

\subsubsection{\textit{Pop} and \textit{Expand}: Greedy Selection}
The \textit{Pop} operation selects a node from the subgraph for the next round expansion iteration, and the \textit{Expand} operation expands from a node by collecting all the related edges of this node.
% The \textit{Pop} operation selects a node with high priority from the subgraph in graph expansion, while \textit{Expand} collects all edges related to the selected node.
% In TRacer, the priority of the nodes is residual, where the node with a higher residual has a greater priority.
Note that the graph expansion in TRacer terminates when the residual of all nodes in the subgraph is below a threshold $\epsilon \in (0,1)$, i.e.,
\begin{equation}
    \max\limits_{u \in V}(\sum\limits_{t \in \Gamma}\sum\limits_{b \in B}r_s(u,t,b)) < \epsilon.
    \label{equ:end_cond}
\end{equation}
Thus we proposed the \textbf{greedy selection} for the \textit{Pop} operation to achieve the condition in Equation \ref{equ:end_cond}, i.e., the \textit{Pop} operation selects the node in the subgraph with the highest residual.
% Thus we proposed two types of node selection rules for the \textit{Pop} operation to achieve the condition in Equation \ref{equ:end_cond}, namely greedy selection, and group selection:
% \begin{itemize}
%     \item \textbf{Greedy selection}: Select one node in the subgraph with the highest residual.
%     \item \textbf{Group selection}: Select all nodes in the subgraph where the residual of every node is greater than $\epsilon$.
% \end{itemize}
% In our implementation, the greedy selection is chosen for \textit{Pop} operation enabled to reduce the pressure of data access in \textit{Expand}.

% \subsection{\textit{Extract}: Local Community Detection}
\subsection{Local Community Detection}
Referring to Figure~\ref{fig:TRacer}, the \textit{Extract} operation employs local community detection to construct a small-scale local community from the expanded graph. For a particular risky source node, the importance rank of the nodes within the local community is significantly higher than the external nodes, making it easy for further expert auditing.

The less conductance \cite{andersen2006local, andersen2007local} of the local community, the higher rank of the nodes in the local community than external nodes, in which the conductance is:
\begin{equation}
    \Phi(S)=\frac{\bm{p}_s(\partial(S))}{\bm{p}_s(S)},
\end{equation}
where $S$ denotes the nodes of the local community, the boundary $\partial(S)=\{ v| (u,v) \in E \land u \in S \land v \in \bar{S} \}$, $\bar{S}$ is the complement of $S$ and $\bm{p}_s(S)=\sum_{u \in S}\bm{p}_s(u)$ denotes the sum of rank over each node in the local community.
Given a specific threshold $\varphi > 0$ for conductance, the local community detection finds the local community satisfying:
\begin{equation}
    \Phi(S) < \varphi.
    \label{equ:local_comm_cond}
\end{equation}
With $S=\{s\}$ as the initialization, Algorithm \ref{alg:ttr_local_comm} describes how to find the local community satisfied the Equation \ref{equ:local_comm_cond}.
\begin{algorithm}[t]
    \caption{TTR-based Local Community Detection}
    \label{alg:ttr_local_comm}
    \begin{algorithmic}[1]
        \REQUIRE The source node $s$, the subgraph of graph expansion $G_s=(V_s, E_s)$, and the TTR score vector $\bm{p}_s$.
        \ENSURE The local community with nodes set $S$.
        \STATE $S = \{ s \}$
        \STATE $\bar{S} = V_s \setminus S$
        \WHILE{$\Phi(S) \geq \varphi$}
            \STATE $u =\mathop{\arg\max}_{v\in\bar{S}}\bm{p}_s(v)$
            % \STATE $u =$ the node with the highest rank in $\bar{S}$
            \STATE $S = S \cup \{ u \}$
            \STATE $\bar{S} = \bar{S} \setminus \{ u \}$
        \ENDWHILE
        \RETURN $G_s.subgraph(S)$
    \end{algorithmic}
\end{algorithm}
\subsection{Theoretical properties}
In this part, we discuss the theoretical properties of TRacer, and prove that our method is able to finish the transaction tracing task in large-scale transaction graphs with a constant time cost. Besides, we discuss the upper limit of the tracing depth in TRacer.

As the description in Proposition \ref{pot:cost}, the cost of TRacer is independent of the graph size, indicating that TRacer is able to trace on the large-scale transaction graph with a low cost.
\begin{proposition}
\label{pot:cost}
The iteration of graph expansion runs $O(\frac{1}{\epsilon \alpha})$ times, and the number of nodes with non-zero values in the output TTR score is at most $O(\frac{1}{\epsilon \alpha})$, which guarantees the cost of local community detection is $O(\frac{1}{\epsilon \alpha})$.
\end{proposition}
\begin{proof}
This follows from Andersen et al. \cite{andersen2006local}, Lemma 2.
\end{proof}

In addition, what depth can TRacer trace in a transaction graph from the source node is described in Proposition \ref{pot:max_depth}.
% where the higher order of the neighbor, the longer distance between the source node and this neighbor.
\begin{proposition}
\label{pot:max_depth}
An $n$-hop neighbor of the source node can be found in the graph obtained by graph expansion, in which $n$ satisfies:
\begin{equation}
    n \leq \frac{log(\epsilon)}{log(1-\alpha)} + 1.
\end{equation}
\end{proposition}
\begin{proof}
In order to maximize the rank of the neighbors far away from the source node, $\alpha$ needs to be as small as possible, and $\beta$ needs to be as close as possible to 0 or 1, which ensures that the residual can be pushed to a specific direction.
Let the sum of residual pushed from the source node $s$ to the $n$-hop neighbors be $r^{(n)}$.
Considering $\beta=1$ here, $r^{(n)}$ can be obtained by:
\begin{equation}
    \begin{cases}
        & r^{(1)}=(1-\alpha), \\
        & r^{(2)} \geq (1-\alpha)(r^{(1)}-k_1 \epsilon)=(1-\alpha)^2-(1-\alpha)k_1 \epsilon, \\
        & ...... \\
        & r^{(n)} \geq (1-\alpha)(r^{(n-1)}-k_{n-1} \epsilon) \\
        & \ \ \ \ \ \ \ \ =(1-\alpha)^n-\sum\limits_{i=1}^{n-1}(1-\alpha)^{n-i}k_{n-i} \epsilon, \\
    \end{cases}
\end{equation}
where $k_i$ represents the number of nodes with residual less than $\epsilon$ in the $i$-order neighbors of the source node.
Therefore, $r^{(n)}$ obtains the maximum value when:
\begin{equation}
    \sum\limits_{i=1}^{n-1}(1-\alpha)^{n-i}k_{n-i} \epsilon = 0,
\end{equation}
which means the residual of each $i$-order node is greater or equal than $\epsilon$. 
This situation exists when the source node is in a path-like graph, whose adjacency matrix $A$ satisfies $A(i,i+1)=1$ and other elements are $0$.
Considering the end condition of the local push procedure, the residual of $n$-hop neighbors is:
\begin{equation}
\begin{cases}
    & (1-\alpha)^n < \epsilon\\
    & (1-\alpha)^{n-1} \geq \epsilon \\
\end{cases}
    % r^{(n)}=(1-\alpha)^n \geq \epsilon
    \Rightarrow 
    \frac{log(\epsilon)}{log(1-\alpha)} < n \leq \frac{log(\epsilon)}{log(1-\alpha)} + 1.
    % n \leq \frac{log(\epsilon)}{log(1-\alpha)}.
\end{equation}
\end{proof}
\section{Experiments}
\label{sec:experiments}
In this section, we conduct experiments to evaluate the effectiveness of TRacer on a large-scale real-world dataset. Besides, we conduct case studies with network visualization techniques.

\subsection{Experimental Setups}
% 主要包括参数设置、数据集描述、指标、对比算法

\subsubsection{Dataset}
We contribute a benchmark dataset including 20 transaction tracing cases in the recent 5 years across three account-based blockchains, i.e., Ethereum, Binance Smart Chain, and Polygon. 
These cases are initialized by various illegal activities containing hacker attacks, Rug-pull, and scams which have caused billions of dollars in losses. 
All these cases are reported by blockchain security companies and verified by experts coming from Certik \footnote{https://www.certik.com/}, Peckchield \footnote{https://peckshield.com/}, Chainalysis \footnote{https://www.chainalysis.com/} and so on.
Some statistics of this dataset are shown in Table \ref{tab:dataset}. 
Note that the transaction data in the dataset is obtained from the open APIs\footnote{https://blockscan.com}.
As we can see, the activities of these cases have acrossed millions of blocks, and we have to trace the money flows of the sources among more than 4 billion transactions.

\begin{table}[t]
  \caption{Statistics of the transaction record data related to the cases }
  \label{tab:dataset}
  \begin{tabular}{l|l|l}
    \hline
    \textbf{Field} & \textbf{Description} & \textbf{Number} \\
    \hline
    Source nodes & \tabincell{l}{The source node related to this\\case, such as the hacker account\\ and the scam contract.} & 20 \\ \hline
    Target nodes & \tabincell{l}{A set of target nodes related to\\this case, such as exchange\\wallets and mixing services.}  & 0.87K \\ \hline
    Blocks & \tabincell{l}{The blocks related to these cases.}  & 23.5M \\ \hline
    Transactions & \tabincell{l}{The transactions contained in the\\blocks related to these cases.} & 4.83B \\ \hline
  \end{tabular}
\end{table}

\subsubsection{Compared Methods}
We compare our method with several baseline blockchain transaction tracing methods. For a fair comparison, we use the general framework of TRacer to reproduce the following comparison methods, including:
\begin{itemize}
    \item \textbf{BFS} \cite{zhao2015graph}: Breadth-First Search, which is the first and the most commonly used transaction tracing method.
    \item \textbf{Poison} \cite{moser2014towards}: A kind of taint analysis technology in blockchain transaction tracing. Each output of a transaction with a dirty input is considered to be tainted in this method.
    \item \textbf{Haircut} \cite{moser2014towards}: A kind of taint analysis technology in blockchain transaction tracing. Each output of a transaction with a dirty input is considered to be tainted partially according to the amount value in this method. 
    \item \textbf{APPR} \cite{andersen2006local}: The approximate personalized PageRank algorithms, which can calculate the relevance of nodes in a network to a given source node with an extremely low cost.
    % \item \textbf{TRacer}: The method we proposed, where the graph construction recognizes the transaction patterns, the TTR and greedy selection are used for graph expansion, and the local community detection is applied in Extract.
    % using a novel local push procedure, and we abbreviate the TTR algorithm with tracing tendency, weight pollution, and temporal reasoning strategies as TTR-base, TTR-weight, and TTR-time respectively.
\end{itemize}

The details of implementing the above transaction tracing technologies can be found in our GitHub page \footnote{https://github.com/wuzhy1ng/BlockchainSpider}.

\subsubsection{Experimental Settings} 
% {\color{blue}Based on the experiences, it is enough for BFS and Poison to trace the 2-order neighbors of the source node, since the increase of order leads to the exponential growth of the size of subgraph in graph expansion for these two methods, bringing great difficulty to transaction auditing.}
%%  !!!! 这段话大有问题
Since the increase of depth can lead to the exponential growth of the size of output graph in BFS and Poison, bringing great difficulty to transaction auditing. During our experiments, we set the upper limit of the tracing depth in these two methods is 2. In addition, we use the Haircut method to trace the ``dirty money'' from the source node until the amount proportion of ``dirty money" of all nodes is less than 0.1\% of that from the source node.
Moreover, we set $\alpha=0.15$, $\epsilon=10^{-3}$ for APPR and TTR, and $\varphi=10^{-3}$, $\beta=0.7$ for TTR to ensure that our method is able to find the paths among the source node and the target nodes with 42-hop at most, according to Proposition \ref{pot:max_depth}.

\subsubsection{Metrics}
we report the average of the following metrics in all cases to measure the effectiveness of transaction tracing:
% We use the following metrics to measure the effectiveness of transaction tracing:
\begin{itemize}
    \item \textbf{Recall}: The recall evaluates how many target nodes can be traced by a method, which is defined as: $Recall = \frac{|V_t|}{|\bar{V_t}|}$, where $|V_t|$ is the number of traced target nodes and $|\bar{V_t}|$ is the number of all target nodes in a case.
    
    \item \textbf{Number of nodes}: This metric measures the number of nodes in the output graph of a case. A smaller output graph with recall ensured is easier for expert auditing.
    % Ensuring the recall, the fewer node number, the more conducive to expert verification.
    
    \item \textbf{Tracing depth}: This metric measures how deep can a transaction tracing method traverse the transaction graph from a source node, indicating that up to $K$-hop neighbors of the source node are detected.
\end{itemize}

\subsection{Experimental Results}
% 参数敏感性实验
\subsubsection{Scalability vs. Performance}
\begin{figure}[t]
    \centering
    \includegraphics[width=0.65\linewidth,height=3.7cm]{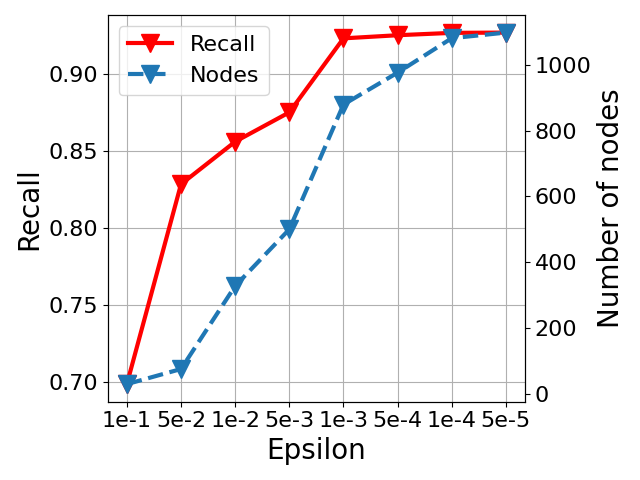}
    \vskip -0.1in
    \caption{The relationship among Epsilon, Recall, and Node number. Note that the recall has reached 70\% with $\epsilon = 10^{-1}$ merely, and the increment of recall becomes slow when $\epsilon$ is less than $10^{-3}$.}
    \label{fig:epsilon_metrics}
\end{figure}
The approximation parameter $\epsilon$ is an important hyper-parameters in modulating the scalability.
To examine the effect of $\epsilon$ on the performance, we repeat experiments with different values of $\epsilon$ and report the recall as well as the number of traced nodes. 
As Figure \ref{fig:epsilon_metrics} shows, the recall has reached 70\% when $\epsilon = 10^{-1}$.
When $\epsilon$ is less than $10^{-3}$, the increment of recall becomes slow, and the number of nodes rapidly increases. 
Therefore, setting $\epsilon = 10^{-3}$ can guarantee a higher recall and fewer nodes with a relatively low cost in experiments, and we set $\epsilon = 10^{-3}$ for TRacer.
% , which is the reason why we set $\epsilon = 10^{-3}$ for TRacer.

% 对比实验
\subsubsection{Comparative experiment}
\begin{table}[t]
  \caption{
  Performance comparison between baselines and TRacer
%   Comparative experiment results. Besides the number of tracing nodes being more than APPR slightly, the recall and tracing depth are better than other methods significantly.
  }
  \label{tab:compared_methods}
  \begin{tabular}{cccc}
    \toprule
    Methods & Recall (\%) & Number of nodes (K) & Tracing depth\\
    \midrule
    BFS & 77.02 & 52.50 & 2.00\\
    Poison & 70.06 & 41.45 & 2.00\\
    Haircut & 58.85 & 10.35 & 4.15\\
    APPR & 71.92 & \textbf{0.66} & 3.60\\
    \textbf{TRacer} & \textbf{92.31} & 0.87 & \textbf{5.05}\\
  \bottomrule
\end{tabular}
\end{table}
Table \ref{tab:compared_methods} shows the performance of different methods, from which we can obtain the following observations.
\textbf{\textit{Firstly}}, the output graphs of BFS and Poison contain an extremely large number of nodes, which brings great difficulty to transaction auditing even through detecting more than 70\% target nodes.
\textbf{\textit{Secondly}}, the output graph of Haircut has fewer nodes than BFS and Poison with a greater tracing depth, but the recall is too low to achieve effective tracing.
\textbf{\textit{Thirdly}}, fewest nodes are obtained by APPR, ensuring the detection results can be easier audited by experts. 
\textbf{Fourthly}, TRacer obtains better performance than APPR. On the basis of the advantages of APPR, the recall and tracing depth of TRacer are significantly better than other methods.

% 消融实验
\subsubsection{Ablation experiment}
\begin{table}[t]
  \caption{Ablation experiment.}
  \label{tab:ablation_experiment}
  \setlength{\tabcolsep}{0.1mm}{
  \begin{tabular}{p{2.5cm}<{\centering}p{1.4cm}<{\centering}p{2.2cm}<{\centering}|p{0.6cm}<{\centering}p{1.8cm}<{\centering}}
    \toprule
    Graph construction with DeFi patterns & Graph expansion & Local community detection & Recall (\%) & Number of nodes (K)\\
    % pattern recognition & expansion & detection & (\%) & of nodes (K)\\
    \midrule
    $\surd$ & $\surd$ & $\surd$ & 92.3 & 0.87\\
     & $\surd$ & $\surd$ & 80.3 & 0.55\\
    $\surd$ & $\surd$ &  & 95.9 & 57.5\\
  \bottomrule
\end{tabular}}
\end{table}
TRacer is consist of three modules including graph construction, graph expansion, and local community detection. 
In order to discuss the function of different modules, we conduct an ablation study and report the performance of TRacer in Table \ref{tab:ablation_experiment} after removing the DeFi pattern recognition in graph construction and the local community detection.
% remove the graph construction module recognizing the edge patterns or local community detection module extracting the important local community, and report the performance as shown in Table \ref{tab:ablation_experiment}.
When the DeFi pattern recognition is removed in the graph construction module, the recall decreases significantly, which shows that the understanding of DeFi patterns in TRacer can help trace the money flows effectively.
Moreover, if the local community detection module is removed, the number of nodes increases significantly with the weakly improvement of recall, which shows that the local community detection module can find the nodes strongly associated with the source node at the cost of a small recall loss.

% Top N recall
\subsubsection{Top\textit{N} Recall}
\begin{figure}[t]
    \centering
    \includegraphics[width=0.6\linewidth]{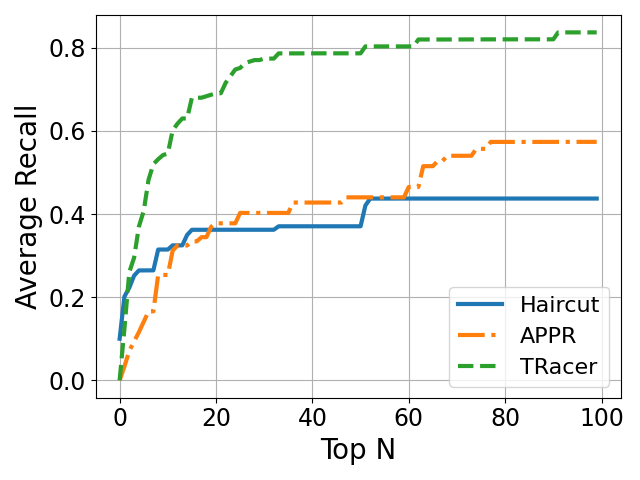}
    \caption{Top $n$ most relevant nodes and the recall. 
    % TRacer tends to give higher importance to the target nodes.
    }
    \label{fig:topn_recall}
\end{figure}

Since the rank of a node represents the relevance relationship between this node and the source node, we can audit the nodes according to the descending order of rank. To compare the rank-based methods including Haircut, APPR, and TRacer, we take out the top $n$ most relevant nodes to the source and calculate the recall for different $n$.
The result is displayed in Figure \ref{fig:topn_recall}, where the curve of TRacer shows a better performance than other methods.
In addition, TRacer achieves a 25\% recall gain over APPR when $n>50$.

\subsection{Case Study}
\label{sec:case_study}
% 该部分对案例中的TTR局部交易网络进行分析
% In this part, we visualize the fund flow on cases with verified labels by Gephi 0.9.2 \cite{bastian2009gephi}.
In this part, we visualize the traced money flow of two cases with Gephi 0.9.2 \cite{bastian2009gephi}, in order to evaluate the feasibility of TRacer.
% In this section, we give the analysis of cases with the local tracing network captured by our method through gephi.

\subsubsection{Cryptopia}\label{sec:case_study_Cryptopia}
\begin{figure}[t]
    \centering
    \includegraphics[width=0.8\linewidth]{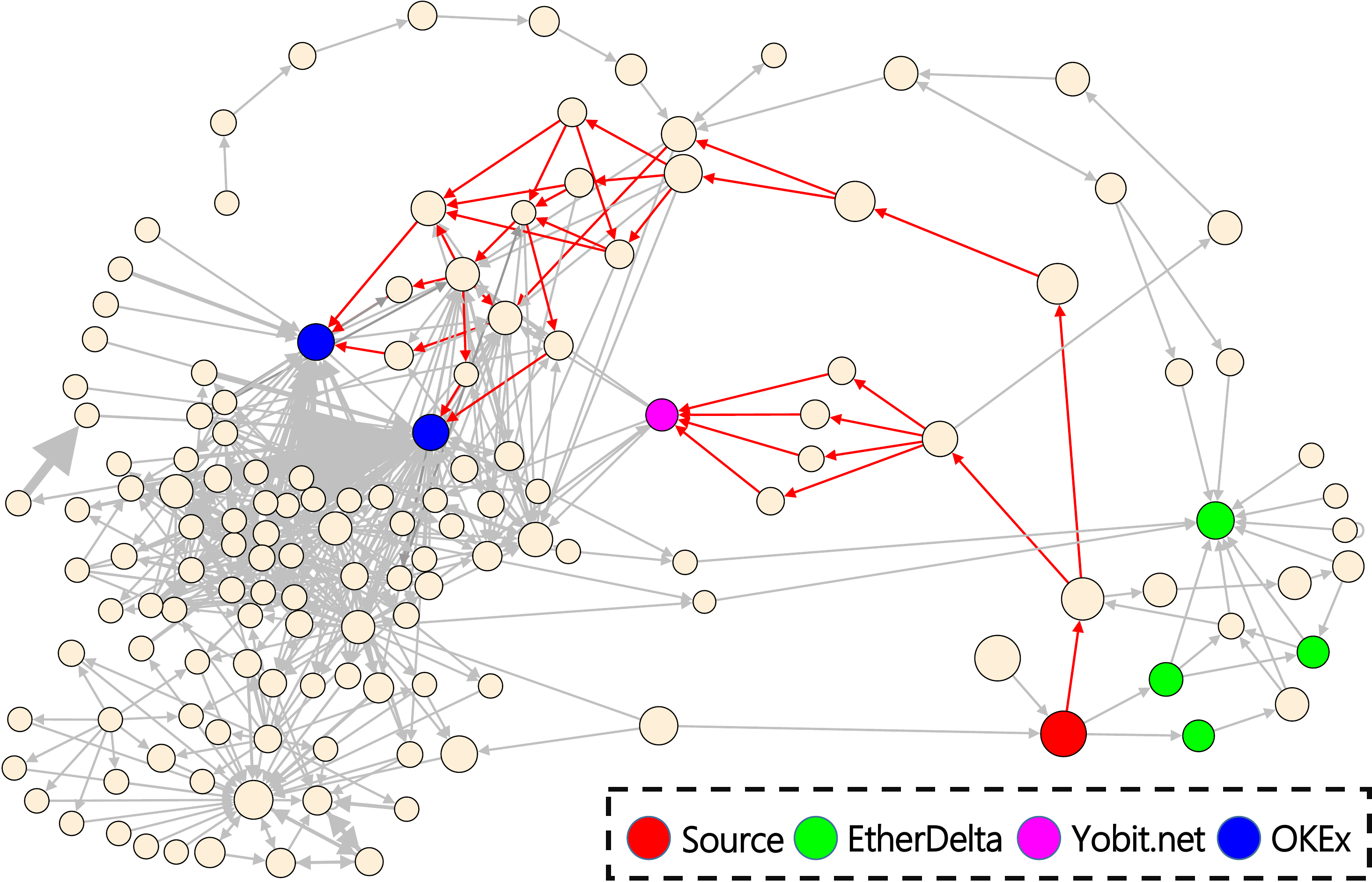}
    \caption{Tracing visualization for Cryptopia. Besides the EtherDelta exchange marked by experts, our method also finds the other target exchanges including Yobit.net and OKEx.}
    \label{fig:ttr_example_cryptopia}
\end{figure}
The Cryptopia exchange was attacked by hackers in May 2019.
According to the tracing result published by CoinHolmes\footnote{https://trace.coinholmes.com}, the source node with prefix 0xd4e79 \footnote{https://etherscan.io/address/0xd4e79226f1e5a7a28abb58f4704e53cd364e8d11} possessed 30.8K stolen ETH from Cryptopia and transferred about 10K ETH to 4 target nodes labeled as EtherDelta.

Figure \ref{fig:ttr_example_cryptopia} presents the transaction tracing result of our method, where the source node and the target nodes are marked with labels. 
Additionally, the node size is proportional to its rank score for each node, so the higher the rank is, the larger the node diameter is.
Based on the tracing result, we can find 4 target nodes labeled as EtherDelta (an exchange) in the 2-hop neighborhood of the source node easily, which is consistent with the results in CoinHolmes.

Moreover, another 4-hop neighbor of the source node labeled as Yobit.net\footnote{https://etherscan.io/address/0xf5bec430576ff1b82e44ddb5a1c93f6f9d0884f3} and two nodes labeled as OKEx can be found in the figure, which is not reported by CoinHolmes.
In fact, Yobit.net and OKEx are exchanges enabling the hacker to cash out the stolen ETH.
% In this way, Yobit.net and OKEx are reasonable to be target nodes.
According to the traced money flows in the figure, about 1420 stolen ETH is transferred into Yobit.net, and 18.47K stolen ETH is transferred into OKEx.
Therefore, more than 97\% of the stolen ETH of Cryptopia are traced by our method in this case. 

% Moreover, another 4-hop neighbor of the source node, labeled as Yobit.net with the address prefix 0xf5bec \footnote{https://etherscan.io/address/0xf5bec430576ff1b82e44ddb5a1c93f6f9d0884f3}, can also be found in the figure, which is not discovered by CoinHolmes.
% In fact, Yobit.net is an exchange enabling the hacker to cash out the stolen ETH, which is reasonable to be a target node.
% And from these fund flow paths, we can find the other 1420 stolen ETH.

% Besides, we can find two nodes labeled as OKEx in the figure, which is the deposit address of the OKEx exchange. 
% According to the traced ETH flows in the figure, it's reasonable that the hackers cashed out another part of the stolen ETH via OKEx, which is about 18.47K ETH.
% Therefore, more than 97\% of the stolen ETH of Cryptopia can be found with our method in this case.

\subsubsection{Kucoin}
\begin{figure}[t]
    \centering
    \includegraphics[width=0.95\linewidth]{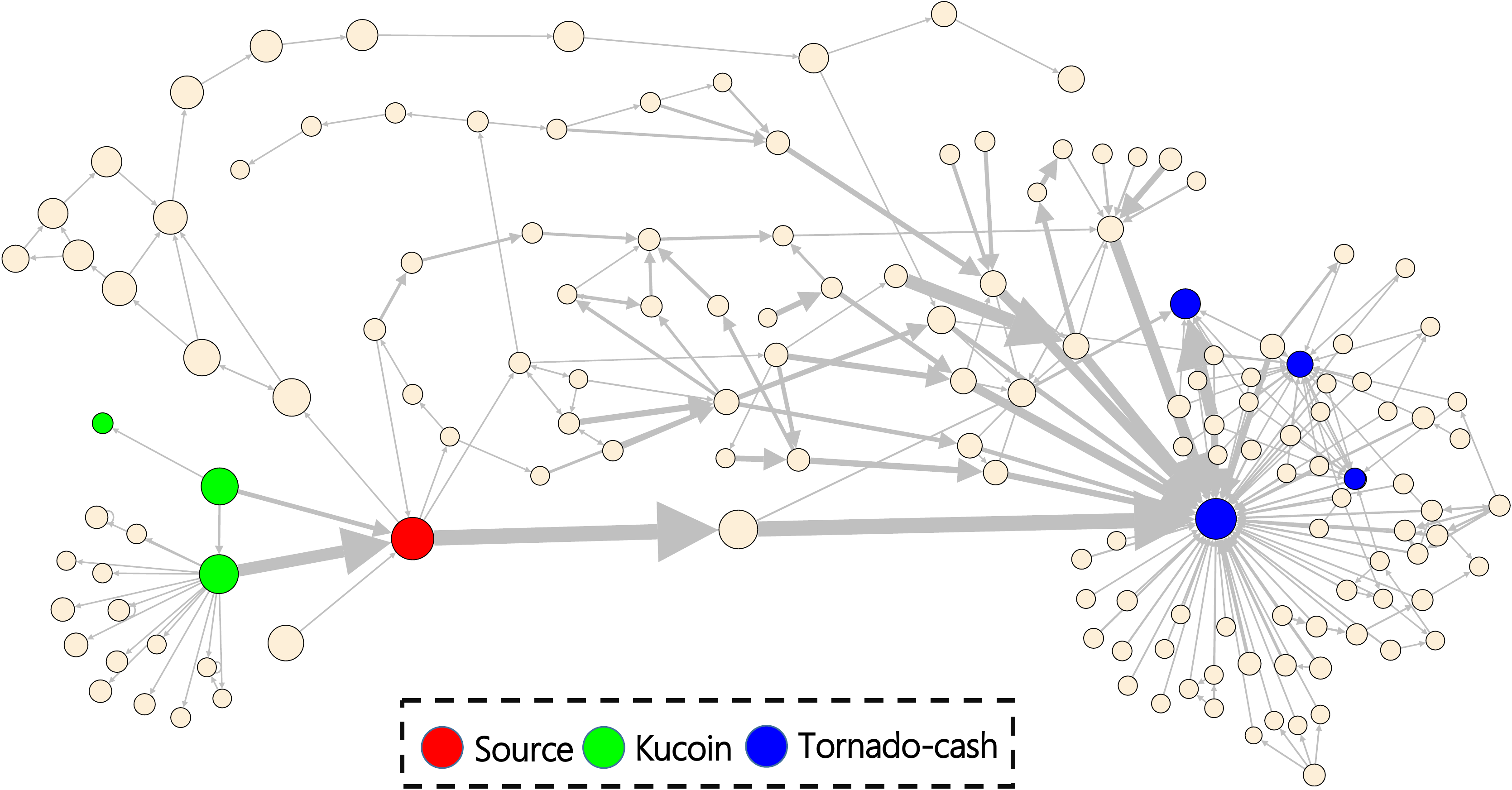}
    \caption{Tracing visualization for Kucoin. A large number of ETH was transferred to Tornado Cash.}
    \label{fig:ttr_example_kucoin}
\end{figure}
The Kucoin exchange was attacked by hackers in September 2020, and there is still a large number of stolen funds that have not been traced until now.
As shown in Figure \ref{fig:ttr_example_kucoin}, the source node \footnote{https://etherscan.io/address/0xeb31973e0febf3e3d7058234a5ebbae1ab4b8c23} of the hackers obtained ETH from the nodes labeled as Kucoin, and then transferred most of the stolen ETH to a famous privacy-preserving protocol named Tornado Cash \cite{tornadocashdoc}.

Through the tracing result of our method, a total amount of 13.8K ETH was transferred to the 100 ETH pool of Tornado Cash, which is similar to the conclusion of experts in SlowMist\footnote{https://coinyuppie.com/uncovering-tornado-cashs-anonymity/}. As Tornado Cash is a decentralized mixing service, it is impossible for us to obtain valuable KYC information from this project to identify the hackers. 
In addition, since Tornado Cash is designed based on zero-knowledge proof, liquidity mining, and smart contract, it is difficult to trace the downstream fund flow when money is transferred into Tornado Cash. Therefore, some researchers have conducted analysis on Tornado Cash in recent years \cite{beres2020blockchain}.
\section{Conclusion}
\label{sec:conclusion}
In this paper, we studied the transaction tracing problem on account-based blockchain trading systems and proposed the first intelligent transaction tracing tool named TRacer. Compared with existing methods based on heuristics and taint analysis, which usually requires expert experience and manual intervention, TRacer show obvious superiority in terms of universality, effectiveness and time cost. In TRacer, we first formulated the transaction records in each account-based blockchain as a directed, weighted, temporal, and multi-relationship graph, which is able to represent the rich semantics of complex multi-token transaction relationships in the system. Then we proposed to use personalized PageRank-based graph searching technologies on this complex graph to trace the money flows. Specifically, we introduced novel approximate personalized PageRank strategies to realize effective and low-cost transaction tracing, which can also handle the complex DeFi transaction actions in account-based blockchains. The theoretical analysis and experimental results demonstrated the effectiveness of TRacer. 
% However, TRacer still cannot trace the fund flow using privacy enhancement technology, especially the mixing services such as Tornado Cash.
In the future, we will further delve into blockchain transaction tracing by integrating more transaction features like bytecodes and logs and design effective methods for blockchain systems with privacy-enhancing mechanisms.
% Moreover, our method can capture the subgraph of the source node, which will be used to carry out more effective node classification in future work.

% \begin{acks}
% \end{acks}

\bibliographystyle{ACM-Reference-Format}
\bibliography{bibfile}

\appendix

\end{document}